\documentclass[hidelinks,onefignum,onetabnum]{siamart251216}

\usepackage{lipsum}
\usepackage{amsfonts}
\usepackage{graphicx}
\usepackage{epstopdf}
\usepackage{algorithmic}
\ifpdf
\DeclareGraphicsExtensions{.eps,.pdf,.png,.jpg}
\else
\DeclareGraphicsExtensions{.eps}
\fi

\usepackage{textcomp}

\usepackage{cite}\usepackage{hyperref}
\usepackage{amsmath,amssymb,amsfonts}
\usepackage{algorithmic}
\usepackage{graphicx}
\usepackage{textcomp}
\usepackage{bm,bbm,mathrsfs,amscd}
\usepackage{calc}
\usepackage{color}
\usepackage{dsfont}
\usepackage{graphicx}
\usepackage{epstopdf}
\usepackage{epsfig}
\usepackage{tikz}
\usepackage{amsmath}
\usepackage{amsfonts}
\usepackage{amssymb}
\usepackage{rotating}
\usepackage{mathtools}
\usepackage{color}
\usepackage{subfig}
\usepackage{enumerate,pgfplots}
\newtheorem{example}{Example}
\newtheorem{remark}{Remark}
\newtheorem{assumption}{Assumption}
\newcommand{\ba}{\begin{array}}
\newcommand{\ea}{\end{array}}

\newcommand{\be}{\begin{equation}}
\newcommand{\ee}{\end{equation}}

\newcommand{\ds}{\displaystyle}

\newcommand{\eps}{\varepsilon}

\newcommand{\mc}{\mathcal}

\newcommand{\ov}{\overline}

\def\1{\boldsymbol{1}}
\def\0{\boldsymbol{0}}

\newcommand{\R}{\mathbb{R}}

\newcommand{\de}{\mathrm{d}}

\newcommand{\tcb}{\textcolor{blue}}
\newcommand{\tcm}{\textcolor{black}}

\def\R{\mathbb{R}}

\tikzstyle{v_c}=[circle, draw,inner sep=2pt, minimum width=12pt, color=blue]
\tikzstyle{v_a}=[circle, draw,inner sep=2pt, minimum width=12pt, color=red]
\tikzstyle{edge} = [draw,thick,-,font=\small ]
\tikzstyle{label} = [draw,fill=black,font=\normalsize]

\newenvironment{proofof}[1]{%
\par\hspace{\parindent}\textit{Proof (of #1):}\ }%
{\hfill$\blacksquare$\par}

\headers{Optimal Control of Behavioral-Feedback SIR Epidemic Model}{M. Alutto, L. Cianfanelli, G. Como, F. Fagnani and F. Parise}

\title{Optimal Control of Behavioral-Feedback SIR Epidemic Model
\thanks{Submitted to the editors 12 June 2026.
	\funding{This work was partially supported by the Research Project PRIN 2022 ``Extracting Essential Information and Dynamics from Complex Networks'' (Grant Agreement number 2022MBC2EZ) funded by the Italian Ministry of University and Research. This material is based upon work supported by the Air Force Office of Scientific Research under award number FA9550-24-1-0082. }
	This work is based in part on material appearing in the first author's PhD dissertation \cite{thesisMartina}. Some of the preliminary results appeared in \cite{CDC2025}.}
	}
	
	\author{Martina Alutto\thanks{Division of Decision and Control Systems, School of Electrical Engineering and Computer Science, KTH Royal Institute of Technology, Stockholm, Sweden (\email{alutto@kth.se}).}
\and Leonardo Cianfanelli\thanks{Department of Mathematical Sciences ``G.L.~Lagrange,'' Politecnico di Torino, 10129 Torino, Italy. (\email{\{leonardo.cianfanelli;giacomo.como,fabio.fagnani\!\}@polito.it}).}
\and Giacomo Como\footnotemark[2]\thanks{Department of Automatic Control, Lund University, 22100 Lund, Sweden.}
\and Fabio Fagnani\footnotemark[2]
\and Francesca Parise \thanks{School of Electrical and Computer Engineering, Cornell University, Ithaca, NY 14850 USA. (\email{fp264@cornell.edu}).}
}
\usepackage{amsopn}

\ifpdf
\hypersetup{
pdftitle={Optimal Control of Behavioral-Feedback SIR Epidemic Model},
pdfauthor={M. Alutto, L. Cianfanelli, G. Como, F. Fagnani and F. Parise}
}
\fi

\begin{document}

\maketitle

\begin{abstract}
	We consider a \textit{behavioral-feedback} SIR epidemic model, in which the infection rate depends in feedback on the fractions of susceptible and infected agents. Contrary to the classic SIR dynamics, the considered model allows one to account for endogenous adaptation mechanisms of the agents in response to the epidemics, such as voluntary social distancing, or the adoption of face masks. For this model, we formulate an optimal control problem for a social planner that has the ability to reduce the infection rate to keep the infection curve below a certain threshold within an infinite time horizon, while minimizing the intervention cost. Based on the dynamic properties of the model, we prove that, under quite general conditions on the infection rate, the \emph{filling-the-box} strategy is the optimal control. This strategy consists in letting the epidemics spread without intervention until the threshold is reached, then applying the minimum control that leaves the fraction of infected individuals constantly at the threshold until the reproduction number becomes less than one and the infection naturally fades out. 
	Our contribution enhances the understanding of epidemic management while accounting for adaptive human behavior, offering novel insights for robust containment strategies.
\end{abstract}

\begin{keywords}
	Epidemic models, Behavioral Feedback SIR model, Optimal control problem.
\end{keywords}

\begin{MSCcodes}
	34H05, 93C10, 49Nxx
\end{MSCcodes}

\section{Introduction}\label{sec:introduction}
The spread of infectious diseases has always been a serious threat to public health. Designing effective containment strategies is however challenging because it requires taking into account several factors, such as the capacity of the healthcare system, the social and economic consequences of interventions, and, most importantly, the adaptive human behavior. When resources are limited, it also becomes important to identify policies that use them in the most effective way. 

In recent years, mathematical models have become important tools for predicting how outbreaks evolve and for assessing the effectiveness of different control strategies. Deterministic compartmental models play a central role in capturing the essential features of disease transmission and supporting the formulation of public health strategies. Among these, the Susceptible-Infected-Recovered (SIR) model is one of the most extensively studied due to its simplicity and generality \cite{Kermack.McKendrick:1927, Hethcote2000TheMO, MKendrickApplicationsOM, Diekmann2000MathematicalEO}. In this model, the population is divided into three compartments (susceptible, infected, and recovered) and the evolution in time of the fraction of agents in each compartment is described through a system of nonlinear differential equations, assuming that the population is large, homogeneous, and well mixed.

However, a key limitation of the classical SIR model is its assumption of a constant infection rate, disregarding behavioral adaptations. Historical and recent evidence shows that agents change their behavior in response to perceived infection risk, affecting contact patterns and, consequently, the disease transmission. For example, during the 1918 influenza pandemic people avoided crowded areas \cite{crosby2003america}, during the SARS outbreak, widespread use of face masks and reduced travel were observed \cite{lau2005sars}, and during the recent COVID-19 pandemic behavioral responses were further reinforced by government measures such as social distancing and lockdowns \cite{hsiang2020effect}. 

To incorporate human behavior into epidemic models, two main approaches have emerged. The first one consists in adding to the system new dynamic variables, often based on \emph{evolutionary game theory} \cite{sandholm2010population}, to model how the agents weigh the costs and benefits of actions like vaccination or isolation based on their risk perception and on social influence. The infection rate is then assumed to depend on such variables, which in turn depend on the epidemic state \cite{satapathi2022epidemic, ye2021game, frieswijk2022meanfield, paarporn2023sis, amaral2021epidemiological, kabir2020evolutionary, martins2023epidemic, certorio2022epidemic}. The drawback of this approach is that the number of variables of the system increases, making optimal control problems often analytically intractable.

The second main modeling approach incorporates the human behavior through static feedback mechanisms. In these \emph{behavioral-feedback SIR epidemic models}, the infection rate $\beta$ is a function of the current epidemic state rather than dependent on additional dynamic variables or constant as in the classical SIR model ~\cite{funk2010modelling, verelst2016review}. A seminal contribution for this class of models is \cite{CAPASSO197843}, which introduced a SIR model where the infection rate $\beta$ decreases with the fraction of infected agents $y$, producing a unimodal infection curve. Later studies confirmed that such feedback can flatten the curve and preserve key stability properties \cite{Baker2020, franco2020feedback, korobeinikov2006lyapunov}, and that the unimodality of the infection curve is preserved even if the infection rate depends on both the fraction of susceptible agents $x$ and infected agents $y$ \cite{Alutto2021OnSE}. Recent work further explored these effects, deriving bounds on epidemic overshoot and identifying invariants of motion \cite{nguyen2024upper}, or studying complex nonlinear feedback with saturation effects \cite{srivastava2024nonlinear}. 
Other behavioral-feedback SIR models with non-linear infection rate $\beta(x,y)$ are considered in \cite{feng2000endemic, gao2024final, liu1987dynamical}. Time-dependent infection rates have also been introduced to reflect seasonal or delayed responses \cite{boatto2018sir}. None of these works, however, studied whether and how optimal control insights derived for the standard SIR model change due to behavioral feedback.

To fill this gap, in this paper, we address an optimal control problem for a class of behavioral-feedback SIR models where the infection rate depends on both the fraction of susceptible agents $x$ and infected agents $y$.
Optimal control \textit{of the classical SIR epidemic} model has been extensively studied in recent years \cite{zino2021analysis, yi2022edge,morton1974optimal,behncke2000optimal,parino2024optimal,Miclo.ea:2022,kruse2020optimal,cianfanelli2021lockdown,acemoglu2021optimal,avram2022optimal,molina2022optimal,arnouss2026adaptive}, particularly for calibrating lockdowns to achieve epidemic control without excessively harming the economy. For example, \cite{kruse2020optimal} proposes an optimal control problem for the SIR model over a finite time horizon, where the epidemic cost is linear in the fraction of infected agents, and the economic cost is linear in the infection rate, proving that the optimal control is bang-bang with at most two switches.
In \cite{cianfanelli2021lockdown}, the authors study an optimal control problem where the economic cost depends on both the severity and duration of the lockdown, and the epidemic cost is quadratic in the fraction of infected agents to account for hospital congestion. They prove that, for an infinite time horizon, a strategy that stabilizes the fraction of infected agents outperforms one that stabilizes the reproduction number. In \cite{acemoglu2021optimal}, the authors solve an optimal control problem numerically for network SIR dynamics, showing that policies that target specific subpopulations such as older individuals significantly outperform uniform policies. 
Several other contributions have investigated related optimal control problems in the classical SIR framework. In \cite{molina2022optimal}, the authors study the problem of minimizing the epidemic peak under a budget constraint, deriving an optimal feedback control for the standard SIR model. In \cite{avram2022optimal}, the authors consider a finite-time optimal control problem for the SIR model under capacity and terminal constraints. More recently,  \cite{arnouss2026adaptive} proposes an adaptive allocation of isolation resources aimed at flattening the infection curve. 
Among this literature, the closest contribution to ours is \cite{Miclo.ea:2022}, where the authors provide the solution of an optimal control problem of the SIR model with infinite time horizon, where an external controller aims to minimize the cost of the intervention while ensuring that the fraction of infected agents remains below a specified threshold at every time. 
This threshold represents the healthcare system’s capacity, referred to as the ICU capacity constraint, ensuring comprehensive care for infected patients. Such constraints reflect scenarios like those observed during COVID-19, where the number of critically ill patients risked exceeding available resources, potentially leading to adverse outcomes. Respecting the ICU constraint guarantees that all patients receive necessary medical care.
The optimal strategy consists in setting the minimum lockdown level needed to satisfy the constraint, known as the \emph{filling-the-box strategy}. In the first stages, the disease spreads uncontrolled, followed by an abrupt lockdown when the infection curves meets the threshold, and then a gradual reopening until a time after which the spread is no longer regulated {and the infection naturally fades out}. The economic rationale behind this optimal policy is to avoid shutting down society while ICU resources remain unused. Therefore, whenever the natural spread does not threaten the ICU constraint, it should proceed without intervention. The same policy has been proven to be optimal also in \cite{Acemoglu2021OptimalAT}, where the authors propose an optimal control problem for epidemic mitigation that incorporates both molecular and serology testing.
Importantly, all the contributions above are formulated \textit{within the classical SIR framework}, whose dynamics allows one to exploit the invariant of motions, relating the fractions of susceptible and infected individuals. 

In this paper, we study the optimal control problem considered in \cite{Miclo.ea:2022} for the behavioral-feedback SIR model, under some mild assumptions on the regularity and monotonicity of the feedback mechanism that regulates the infection rate. Our main contribution is to prove that the filling-the-box strategy remains the optimal control \textit{for a general class of behavioral-feedback SIR models} that encompasses the classical SIR model. To the best of our knowledge, this is the first contribution that provides analytical solutions to optimal control problems for the behavioral feedback SIR model, thus extending the findings of a rich strand of literature \cite{Miclo.ea:2022,molina2022optimal, avram2022optimal, arnouss2026adaptive}.
Compared to these papers, the main technical challenge of our analysis is that the class of behavioral feedback SIR models considered in this paper does not admit in general an invariant of motion, which was exploited in previous works to prove the optimality of the filling-the-box strategy.
Our proof relies on general properties of the behavioral-feedback SIR model, in particular the unimodality of the infection curve, which are guaranteed under very mild assumptions on the feedback mechanism.

This work extends our preliminary results in \cite{CDC2025}, where we focused on a behavioral-feedback SIR model with a specific feedback mechanism and we derived a non-trivial invariant of motion, characterized the infection peak, and computed the cost of the filling-the-box strategy, without proving its optimality. This work significantly extends such results by considering a broader class of feedback mechanisms, and, most importantly, by proving optimality of the filling-the-box strategy, without the need of an invariant of motion for the analysis. 
Moreover, we provide a numerical counterexample showing that, when the monotonicity assumptions on the infection rate are relaxed, the filling-the-box strategy is no longer optimal.

The rest of the paper is organized as follows. In Section~\ref{sec:behavioral-epidemic-model} we introduce the controlled behavioral-feedback SIR epidemic model, formulate the optimal control problem we aim to study, and present our main result. In Section~\ref{sec:uncontrolled-model}, we analyze the uncontrolled dynamics, as a first instrumental step for the proof of our main result. The proof itself is then provided in Section~\ref{sec:proof-optimality}. In Section~\ref{sec:counterexample}, we discuss a counterexample, while in Section~\ref{sec:conclusion} we summarize our contribution and discuss future research lines.

\subsection{Notation}
We briefly gather here some notational conventions adopted throughout the paper. We denote by $\R$ and $\R_{+}$ the sets of real and nonnegative real numbers, respectively. We will often consider differentiable functions $f(x,y)$ and use the notation $\dot f=f_x\dot x+f_y \dot y$ to denote the time derivative of the composite function $f(x(t), y(t))$. 
Throughout this paper, given an interval $\mc I\subseteq\R$, we say that $u:\mc I\to\R^n$ is a piecewise-continuous function if there exists a subset $\mc W\subseteq\mc I$ that does not contain accumulation points and is such that:
\begin{itemize}
	\item $u$ is continuous on $\mc I\setminus\mc W$;
	\item in every $t_0$ in $\mc W$, $u$ has a jump discontinuity, i.e., the left and right limits 
	$$u(t_0^-)=\lim_{t \to t_0^-}u(t)\,,\qquad u(t_0^+)=\lim_{t \to t_0^+}u(t)\,,$$
	exist and are finite. 
\end{itemize}
A function $u:\mc I\to\R^n$ is right-continuous if it is piecewise-continuous and $u(t_0)=u(t_0^+)$ for every discontinuity point $t_0$. 
Furthermore, we say that a continuous function $u:\mc I \to\R$ is piecewise-$\mc C^1$ if, except for a subset $\mc W$ as above, it is differentiable and $u'$ is piecewise-continuous.

\section{Controlled Behavioral-Feedback SIR epidemic Model}\label{sec:behavioral-epidemic-model} 
In this section, we introduce the epidemic model studied in this paper and provide our main result. 
\subsection{Model Definition}
We build our analysis on an epidemic model that describes the spread of a disease within a homogeneous population of agents, incorporating an endogenous feedback mechanism and an exogenous control signal.
Before introducing the model, we define the simplex $\Delta$ and the space $\mc S$ by
$$\Delta=\left\{(x,y,z)\in\R_+^3:\,x+y+z= 1\right\}\,,$$
and
$$\mc S = \{(x,y) \in \R_+^2:\, x+y \le 1\}\,,$$
respectively. 
The \emph{controlled behavioral-feedback SIR epidemic model (CBF-SIR)}  is the system of ODEs
\be\label{control-system-full}\begin{cases}
	\dot x=-(1-u)\beta(x,y)xy,\\
	\dot y=(1-u)\beta(x,y)xy-\gamma y,\\
	\dot z=\gamma y\,,
\end{cases}\ee
where 
\begin{itemize}
	\item $x= x(t)$, $y=y(t)$ and $z=z(t)$ indicate the fraction of susceptible, infected and recovered agents, respectively, for every time $t \ge 0$;
	\item $\beta : \mc S \to(0,+\infty)$ is the state-dependent infection rate, which captures the endogenous behavioral response to the epidemics, assumed to be a $\mc C^1$ function of $(x,y)$;
	\item $u:\mathbb{R}_+ \to [0,1]$ is a right-continuous function of time that models a control signal available to an external controller aiming to limit the disease spread by restricting social interactions. The set of available control signals is
	$$
	\mc U = \{u: [0,+\infty) \to [0,1] \ \text{right-continuous}\};
	$$
	\item $\gamma>0$ is the recovery rate from the disease.
\end{itemize}
The following result establishes the well-posedness of the CBF-SIR model \eqref{control-system-full} in the simplex $\Delta$.
\begin{proposition}\label{prop:general} 
	Given a state-dependent infection rate $\beta$ of class $\mc C^1$, a control signal $u$ in $\mc U$ with jumps in $t_1<t_2<\cdots$, and an initial state $(x_0, y_0,z_0)$ in $\Delta$, there exists a unique piecewise-$\mc C^1$ function $(x,y,z):[0, +\infty)\to \Delta$ 
	that solves the ODE \eqref{control-system-full} outside the time instants $t_1<t_2<\cdots$, and is such that
	\be\label{initial-state}x(0)=x_0\,,\qquad y(0)=y_0\,,\qquad z(0) =z_0\,.\ee 
	Moreover:
	\begin{enumerate}
		\item[(i)] $t \to x(t)$ is non-increasing and is strictly decreasing if and only if $x_0>0$ and $y_0>0$;  
		\item[(ii)] $y(t)>0$ for every $t \ge 0$ if and only if $y_0>0$;
		\item[(iii)] the set of equilibrium points in $\Delta$ is 
		$$\Delta_e = \{(x_e, 0, 1-x_e): x_e \in [0,1]\};$$
		\item[(iv)] $(x(t), y(t), z(t))$ converges, for $t\to +\infty$, to an equilibrium point in $\Delta_e$.
	\end{enumerate}
\end{proposition}

\begin{proof}
	See Appendix.
\end{proof}
From now on, we will refer to the vector functions constructed in the previous proposition as \emph{solutions} of system \eqref{control-system-full}.

Observe that the map $s : \mc S \to \Delta$ defined by $$s(x,y) = (x,y,1-x-y)\,,$$ is a bijection from $\mc S$ to $\Delta$, hence for every time $t$ we can denote the state by the pair $(x(t),y(t))$ in $\mc S$ and omit $z(t)$, under the implicit assumption that $z(t)=1-x(t)-y(t)$. In light of this observation, from now on, we can limit our analysis to the first two equations in 
\eqref{control-system-full}. 
After defining the \emph{reproduction number}
\be\label{eq:reproduction}R(x,y)=\frac1\gamma\beta(x,y)x\,,\ee
the two equations can be written as
\be\label{control-system}\begin{cases}
	\dot x=-(1-u)\gamma R(x,y)y\,,\\
	\dot y=\gamma((1-u)R(x,y)-1)y\,.
\end{cases}\ee
From now on we will work with the system of ODEs  \eqref{control-system} and refer to it as the CBF-SIR epidemic model.

	
	
	To prove both the unimodality of the infection curve and to further show the optimality of the proposed control strategy, we rely on the following standing assumption on the feedback mechanism regulating the infection rate.
	\begin{assumption}\label{ass:ass1}
		The state-dependent infection rate $\beta$ satisfies
		\begin{align}
			\label{condition2} 
			x\beta_x(x,y)+&\beta(x,y)>0\,,\\[4pt]
			\label{condition3} 
			\beta_y(x&,y) \leq 0\,,
		\end{align}
		for every $(x,y)$ in $\mc S$. 
	\end{assumption}
	\begin{remark}\label{remark:reproduction}
		Note that \eqref{condition2}-\eqref{condition3} are equivalent to assuming that the reproduction number $R(x,y)$ is increasing in $x$ and non-increasing in $y$ for every $(x,y)$ in $\mc S$. 
	\end{remark}
	
	Equation \eqref{condition2} captures how the infection rate varies with the fraction of susceptible agents. This condition is satisfied for example when $\beta_x(x,y)$ is positive (e.g., when a higher fraction of susceptible leads to a lower perceived risk, encouraging social interactions and riskier behaviors that facilitate the transmission of the disease). Equation \eqref{condition2} can also be satisfied when $\beta_x(x,y)$ is negative but not too large in magnitude. A negative value of $\beta_x(x,y)$ can be interpreted as the effect of pandemic fatigue, where agents gradually relax self-protective measures as the epidemic progresses. Indeed, by Proposition~\ref{prop:general}(i), the fraction of susceptible agents $x$ decreases over time, so this condition reflects an increase in the infection rate as time progresses.
	Equation \eqref{condition3} captures the tendency of the agents to reduce their interaction level in response to the number of active cases. We remark that \eqref{condition2} generalizes the assumption made in \cite{Alutto2021OnSE}, while \eqref{condition3} is a standard assumption in the literature of feedback SIR epidemic models \cite{Baker2020, franco2020feedback}. 
	
	\begin{example}\label{ex2} 
		Let $a\ge0$ and $b: \mathbb{R} \to (0,+\infty)$ be an arbitrary non-decreasing $\mc C^1$ function.
		The two families of infection rates
		\be \label{eq:example}\beta(x,y)=\frac{b(x)}{1+ a y }\,, \ee
		\be\label{eq:example2} \beta(x,y) = b(x) (1-ay)\,,\ee
		satisfy \eqref{condition2}-\eqref{condition3}.
		The CBF-SIR epidemic model with infection rate \eqref{eq:example} has been studied in \cite{CDC2025}. Note that also the classical SIR model may be recovered as a special case from \eqref{eq:example}-\eqref{eq:example2} by letting $a=0$ and $b$ being a positive constant.
	\end{example}
	
		
		\subsection{Problem Statement}
		We consider an external controller that designs a control signal  $u$ in $\mc U$ to reduce the infection rate. The goal of the controller is to keep the fraction of infected agents below a threshold $\bar y$ in $(0,1]$, representing for instance the healthcare capacity.
		\begin{definition}
			Given $\bar y$ in $(0,1]$ and $(x_0, y_0)$ in $\mc S$ such that $y_0 \le \bar y$, we say that a control signal $u$ in $\mc U$ is \emph{feasible} if the solution of the CBF-SIR epidemic model \eqref{control-system} with initial condition $(x_0, y_0)$ satisfies $y(t) \le \bar y$ for every $t \ge 0$. The set of all feasible control signals for a given threshold $\bar y$ and initial condition $(x_0,y_0)$ is denoted $\mc U^{\bar y}_{x_0,y_0}$.
		\end{definition}
		
		Note that $\mc U^{\bar y}_{x_0,y_0}$ is non-empty since it contains the trivial control signal $u \equiv 1$.
		
		The implementation of the control has an associated economic cost that the controller aims to minimize. The cost is modeled by the functional
		$J : \mc U^{\bar y}_{x_0,y_0} \to [0,+\infty]$ defined by
		$$J(u) = \int_{0}^{+\infty} u(t) \, dt\,.$$
		We then consider the optimal control problem 
		\be \begin{aligned}
			\label{control-problem} 
			\mspace{-15mu}V^*(x_0,y_0)\mspace{15mu}=\mspace{15mu} & \min_{u \in \mc U_{x_0,y_0}^{\bar{y}}} J(u).
		\end{aligned}
		\ee
		Note that the solution of \eqref{control-problem} is trivial if $y_0=0$, since Proposition~\ref{prop:general}(ii) guarantees that in this case $y(t) = 0$ for every $t \ge 0$ under every control signal in $\mc U$, including the zero cost control $u \equiv 0$. For this reason, we will work under the assumption that $y_0>0$. In the next section we establish our main result, which provides a solution to \eqref{control-problem}.

		\subsection{Main result}
		We now formulate our main result, whose proof is deferred to Section~\ref{sec:proof-optimality}. To this end, let \tcm{$\mc S_+ = \{(x,y) \in \mc S: x>0, y>0\}$} and
		\be\label{rho-def}\rho(x,y)=\frac{R(x,y)-1}{R(x,y)}=1-\frac{\gamma}{x \beta(x,y)} \,,\ee
		for every $(x,y)$ in $\mc S_+$.
		\begin{theorem}\label{theo:theo-control} 
			Let $\beta$ be of class $\mc C^2$ and let Assumption~\ref{ass:ass1} hold true.
			Consider a threshold $\bar y$ in $(0,1]$ and an initial condition $(x_0,y_0)$ in $\mc S$ with $0< y_0 \le \bar y$.
			Then, the optimal control problem \eqref{control-problem} admits a solution $u^*$ in $\mc U_{x_0,y_0}^{\bar{y}}$ that has the following feedback representation. Indicating with $(x^*(t), y^*(t))$ the solution of \eqref{control-system} with initial condition \eqref{initial-state} under the control signal $u^*(t)$, we have that 
			\be\label{feedback}u^*(t) = \mu(x^*(t), y^*(t)),\quad \forall t\geq 0\,,\ee 
			where 
			\be\label{eq:mu}
			\mu(x,y) = \begin{cases}
				0 & \text{ if } y < \bar{y} \\
				\ds\left[\rho(x,\bar y)\right]_+ & \text{ if } y = \bar{y}\,.
			\end{cases} \ee
			Moreover, $u^*$ is unique  in $\mc U_{x_0,y_0}^{\bar{y}}$.
		\end{theorem}
		\begin{proof}
			See Section~\ref{sec:proof-optimality}. 
		\end{proof}
		
		\begin{figure}
			\centering
			\includegraphics[scale=0.6]{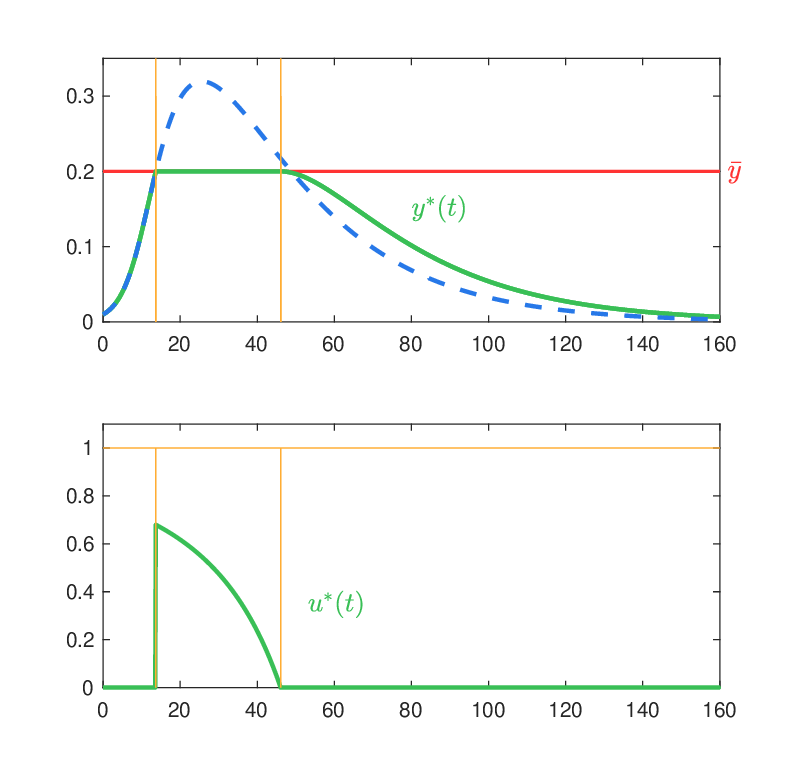}
			\caption[Controlled behavioral-feedback SIR epidemic model]{Numerical simulation of the CBF-SIR epidemic model~\eqref{control-system} with $\beta(x,y)= 0.35 x(1-y)$ and $\gamma=0.05$. The top plot shows in blue the uncontrolled solution and in green the solution corresponding to the optimal control signal $u^*(t)$~\eqref{feedback}-\eqref{eq:mu} with threshold $\bar y = 0.2$. The optimal control signal is reported in the bottom plot.}
			\label{fig:control-policy}
		\end{figure} 
		
		The optimal feedback policy $u^*(t)$ allows the disease to spread without control until the fraction of infected agents hits the threshold $\bar{y}$ at a certain time $T_0$. From there on, the dynamics read
		\be\label{control-system-hit}
		\begin{cases}
			\dot x^* = -\gamma\bar y	\\
			\dot y^* = 0\,,
		\end{cases}
		\ee
		hence the fraction of infected remains constantly at the threshold $\bar y$. At time $T_1$ such that $R(x^*(T_1),\bar y) = 1$ (or equivalently, $\rho(x^*(T_1),\bar y) = 0$), the control is released, so that the trajectory follows the uncontrolled dynamics and the infection curve naturally declines from there on. The corresponding control signal $u^*(t)$ starts abruptly above $0$ at time $T_0$ and subsequently is adjusted in time to maintain the fraction of infected exactly at the threshold $\bar y$ until $T_1$. At that point, the control is deactivated. This policy is known in the literature as \emph{filling-the-box} strategy and its optimality relies on the monotonicity of the reproduction number, which implies the unimodality of the infection curve. 
		Our result generalizes \cite{Miclo.ea:2022}, which proved that the filling-the-box strategy is the optimal control for the classical SIR epidemic model, which is a special case of the CBF-SIR epidemic model, as noted in Example \ref{ex2}.
		
		A numerical simulation of the CBF-SIR epidemic model \eqref{control-system} with $\beta(x,y)= 0.35 x(1-y)$, $\gamma = 0.05$, and optimal control signal $u^*(t)$ \eqref{feedback}-\eqref{eq:mu} with threshold $\bar y = 0.2$ is presented in Figure~\ref{fig:control-policy}. The top plot illustrates in green the solution $y^*(t)$ corresponding to the optimal control signal $u^*(t)$, which is reported in the bottom plot. The dashed blue curve is the corresponding uncontrolled solution, namely, the solution of the CBF-SIR epidemic model \eqref{control-system} with $u \equiv 0$.
		
		The next section analyzes the property of the CBF-SIR epidemic model without control. This analysis is instrumental for the proof of Theorem~\ref{theo:theo-control}.
		
		\section{Uncontrolled BF-SIR epidemic model}\label{sec:uncontrolled-model}
		In this section, we analyze the CBF-SIR epidemic model \eqref{control-system} with $u\equiv0$, i.e., the dynamical system
		\be\label{uncontrol-system}\begin{cases}
			\dot x=-\gamma R(x,y)y\,,\\
			\dot y=\gamma[R(x,y)-1]y\,.\\
		\end{cases}\ee
		
		We shall refer to this system of ODE as the (uncontrolled) BF-SIR epidemic model.
		
		\subsection{Unimodality of the infection curve}
		We start by proving that the infection curve of the uncontrolled BF-SIR epidemic model \eqref{uncontrol-system} under condition \eqref{condition2} is unimodal, namely, it admits at most one peak. This result generalizes \cite[Theorem 1]{Alutto2021OnSE}.
		\begin{proposition}\label{prop:peak} Consider the uncontrolled BF-SIR epidemic model \eqref{uncontrol-system} with initial condition \eqref{initial-state} and let \eqref{condition2} hold true. Then:
			\begin{enumerate}
				\item[(i)] if $R(x_0,y_0)\le1$, then $t\mapsto y(t)$ is strictly decreasing for $t\ge0$;
				\item[(ii)]if $R(x_0,y_0)>1$, then there exists a finite time $\hat t$ such that $t\mapsto y(t)$ is strictly increasing for $t$ in $[0,\hat t]$ and strictly decreasing for $t\ge\hat t$;
			\end{enumerate}
		\end{proposition}
		\begin{proof} 
			It follows from the definition of the reproduction number \eqref{eq:reproduction} and from \eqref{uncontrol-system} that
			\begin{equation}\label{dotR}
				\tcm{\dot R = - R_x \gamma Ry + R_y\gamma (R-1)y\,.}
			\end{equation}
			Note that, for a given time $\hat t$ in which $R(x(\hat t),y(\hat t)) = 1$, this equation, together with \eqref{condition2}, implies that $\dot {R}(x(\hat t),y(\hat t)) < 0$. 
			Hence, if $R(x_0,y_0) \le 1$, it must be $R(x(t),y(t))<1$ for every $t>0$, so that $y(t)$ is necessarily strictly decreasing for every positive time $t \ge 0$, proving item (i). To prove item (ii), let $R(x_0,y_0)>1$ and assume by contradiction that $R(x(t),y(t))>1$ for every $t \ge 0$. Then, $y(t)$ would be strictly increasing for all times and could not converge to zero, contradicting Proposition~\ref{prop:general}(iv). Therefore, $\hat t=\min\{t\geq 0\,|\, R(x(t),y(t))=1\}$ must exist and has to be finite, otherwise the fraction of recovered agents would diverge in time due to the third equation in \eqref{control-system-full}, violating Proposition~\ref{prop:general}. Hence, $y(t)$ is strictly increasing in $[0,\hat t]$. For $t \ge \hat t$, item (i) implies that $y(t)$ is strictly decreasing, proving item (ii).
		\end{proof}
				
				\subsection{Some geometric considerations}
				We analyze the dynamics of the uncontrolled BF-SIR model \eqref{uncontrol-system} in presence of a threshold $\bar y$ in $(0,1]$. Our goal is to determine which trajectories remain below this threshold without requiring external intervention. 	
				
				It follows from \cite[Chapter 2]{thesisMartina} that, for every initial condition $(x_0,y_0)$ in $\mc S$, the uncontrolled BF-SIR model \eqref{uncontrol-system} admits a unique solution of class $\mc C^1$ in $\mc S$. Consequently, there exists a map $\phi:\R_+\times\mc S\to\mc S$ such that $(x_0, y_0)\mapsto \phi(t,x_0,y_0)$ and $t\mapsto \phi(t, x_0, y_0)$ are both of class $\mc C^1$ and $t\mapsto \phi(t, x_0, y_0)$ is the unique solution of the uncontrolled BF-SIR model \eqref{uncontrol-system}  with $\phi(0, x_0, y_0)=(x_0,y_0)$ in $\mc S$. The positive semi-orbit for the initial condition $(x_0,y_0)$ is defined as $$\Gamma^+(x_0,y_0) = \{\phi(t,x_0,y_0), t \in \R_+\}\,.$$
				For every state $(x,y)$ in $\mc S$, we also define the set
				$$\Gamma^-(x,y) = \{(x_0,y_0)\in\mc S:\, (x,y)\in\Gamma^+(x_0,y_0)\}\,,$$
				of initial conditions in $\mc S$ whose positive semi-orbit contains $(x,y)$. With a slight abuse of terminology, we shall refer to $\Gamma^-(x,y)$ as the negative semi-orbit of $(x,y)$ for the uncontrolled BF-SIR model \eqref{uncontrol-system}. \smallskip
				
				\begin{figure}
					\centering
					\includegraphics[scale=0.5]{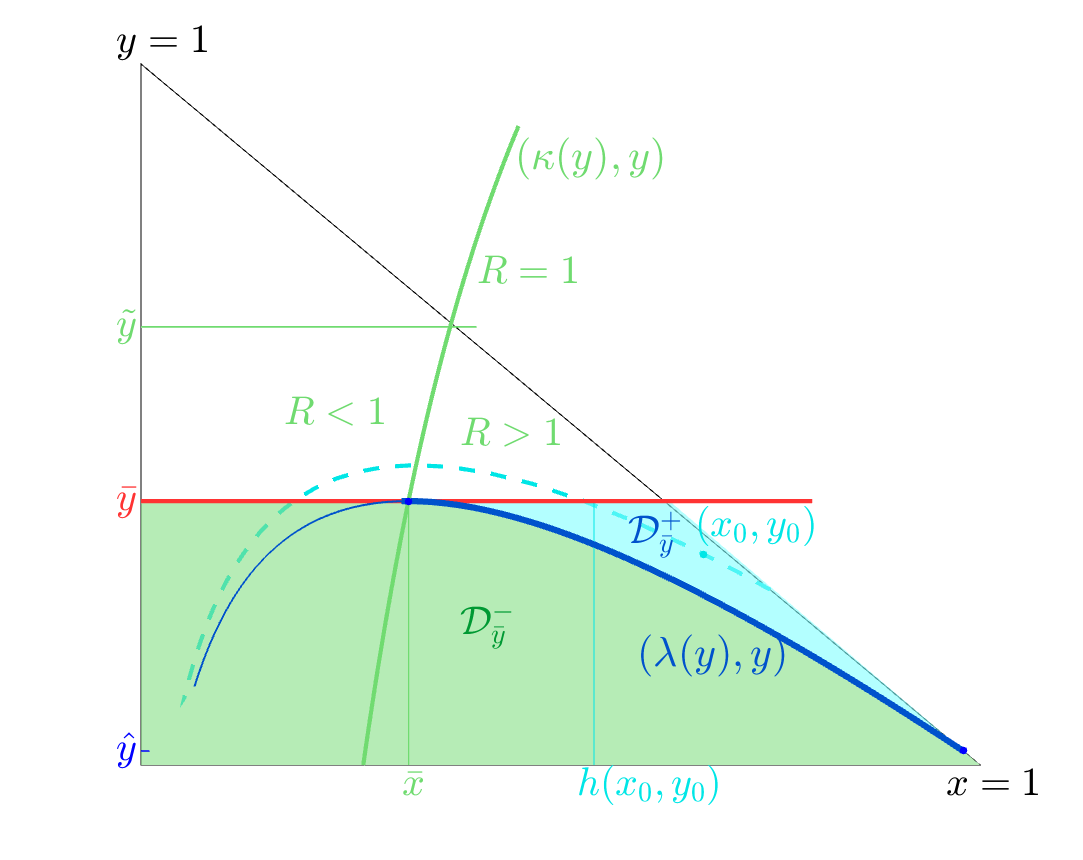}
					\caption{Phase portrait of the uncontrolled BF-SIR model \eqref{uncontrol-system} on the state space $\mc S$, with recovery rate $\gamma=0.05$ and infection rate \eqref{eq:example2} with $b(x)=(x+2)/10$ and $a =0.5$. 
					}
					\label{fig:phase}
				\end{figure}
				
				Let $\mc D_{\bar y}$ denote the set of states that the trajectory can visit without violating the threshold constraint, that is
				$$
				\mc D_{\bar y} = \{(x,y) \in \mc S: 0 < y \le \bar y\}\,.
				$$
				The first part of our analysis leads to a decomposition of $\mc D_{\bar y}$ into two sets $\mc D_{\bar y}^-$ and $\mc D_{\bar y}^+$, which will be formally defined after Lemma~\ref{lemma:lambda}. Informally, $\mc D_{\bar y}^-$ contains the initial conditions from which the fraction of infected remains below the threshold $\bar y$ without requiring an external intervention and is the largest positively invariant set contained in $\mc D_{\bar y}$, as established by Lemma~\ref{lemma:D0}. In contrast, $\mc D_{\bar y}^+ = \mc D_{\bar y} \setminus \mc D_{\bar y}^-$ contains all initial conditions such that the solution of the uncontrolled dynamics would exceed the threshold $\bar y$, thus requiring an external intervention.
				The two regions are  illustrated in Figure~\ref{fig:phase}. 
				
				The following statement characterizes the set of states such that $R(x,y)=1$ and is instrumental for the definition of such sets.
				\begin{lemma}\label{lemma:kappa} 
					Let Assumption~\ref{ass:ass1} hold true and assume that
					\be\label{Rmax>1} \beta(1,0)>\gamma\,.\ee
					Then, there exists a unique $\tilde y > 0$  such that $R(1-\tilde y,\tilde y) = 1$ and a $\mc C^1$ function $\kappa : [0,\tilde y] \to [0,1]$ such that
					$$\{(x,y)\in\mc S:\,R(x,y)=1\}=\{(\kappa(y),y):\,y\in [0,\tilde y] \}\,.$$
				\end{lemma}
				\begin{proof} 
					Note from \eqref{eq:reproduction} that $R(0,1)=0$, while $R(1,0)>1$ due to \eqref{Rmax>1}. Therefore, since $y\mapsto R(1-y,y)$ is continuous on $[0,1]$, there exists $\tilde y$ in $(0,1)$ such that $R(1-\tilde y,\tilde y)=1$. 
					For every $y \in [0,\tilde y]$, note that $R(0,y) = 0$ by \eqref{eq:reproduction}. 
					From \eqref{condition2}-\eqref{condition3} it follows that $R$ is increasing in $x$ and non-increasing in $y$. Since $R(1-\tilde y,\tilde y) = 1$, for any $y \leq \tilde y$ one has $1-y \geq 1-\tilde y$, which implies $R(1-y,y) \ge R(1-\tilde y,\tilde y) =1$.
					Therefore, \eqref{condition2} implies the existence of a unique $x = \kappa(y)$ such that $R(\kappa(y),y)=1$. The fact that $\kappa$ is of class $\mc C^1$ follows from the Implicit Function Theorem~\cite[Theorem 9.18]{rudin1976principles} and from \eqref{condition2}.
				\end{proof}
				
				\begin{remark}
					The function $\kappa$ maps every $y$ in $[0,\tilde y]$ into the unique fraction of susceptible agents $\kappa(y)$ in $[0, 1-y]$ such that $R(\kappa(y),y)=1$. In the classical SIR model,  $\tilde y = 1 - \gamma/\beta$ and $\kappa(y) = \gamma/\beta$ for every $y$ in $[0, \tilde y]$, which corresponds to the \textit{herd immunity threshold}. 
				\end{remark}
				
				Figure~\ref{fig:phase} illustrates in green the curve  $x = \kappa(y)$. Due to the monotonicity of the reproduction number, this curve splits $\mc S$ in two areas, the left-hand one with $R(x,y)<1$ and the right-hand one with $R(x,y)>1$. 
				
				Now, let $\bar x=\kappa(\bar y)$. It is useful in our analysis to consider the negative semi-orbit $\Gamma^-(\bar x,\bar y)$ for the uncontrolled BF-SIR model \eqref{uncontrol-system}. The following result gives information on $\Gamma^-(\bar x,\bar y)$ and permits the definition of $\mc D_{\bar y}^+$ and $\mc D_{\bar y}^-$.
				\begin{lemma}\label{lemma:lambda}\label{lemma:negorb} Let Assumption~\ref{ass:ass1} and \eqref{Rmax>1} hold true and assume that $\bar y \le \tilde y$. Then, there exists $\hat y$ in $[0, \bar y]$ and a $\mc C^1$ decreasing function $\lambda:[\hat y,\bar{y}] \rightarrow [\bar{x},1]$ such that:
					\begin{enumerate}
						\item[(i)] 
						$\ov{\Gamma^-(\bar x,\bar y)}=\{(\lambda(y),y):\,\hat y\leq y\le\bar y\}\,;$
						\item[(ii)]  there exists $(x,y)$ in $\Gamma^{-}(\bar x,\bar y)$ with $x+y=1$ if and only if $\hat y>0$. In such a case, $\hat y$ is the unique element in $(0,\bar y]$ such that $\lambda(\hat y)+\hat y=1$;
						\item[(iii)] $\lambda(y) > \kappa(y)$ for every $y$ in $[\hat y,\bar{y})$ and $\lambda(\bar y)=\kappa(\bar y)=\bar x$.
						In particular, $R(\lambda(y),y)>1$ for every $[\hat y,\bar{y})$.
					\end{enumerate}
				\end{lemma}
				\begin{proof}
					We define the \textquotedblleft backwards" dynamics
					\be\label{eq:backward}
					\begin{cases}
						\dot x = \beta(x,y)xy,\\
						\dot y = - \beta(x,y)xy + \gamma y\,,
					\end{cases}
					\ee
					and let $(x(t),y(t))=\phi^-(t,\bar x,\bar y)$ denote the unique solution of \eqref{eq:backward} at time $t$ in $\R_+$ with initial condition $(x(0),y(0)) = (\bar x,\bar y)$.
					Note that \eqref{eq:backward} is obtained from \eqref{uncontrol-system} by letting $t \to -t$, so that 
					\be\label{eq:symmetry}
					\Gamma^-(\bar x,\bar y) = \{\phi^-(t,\bar x,\bar y): t \in \R_+\}\,.
					\ee 
					Note that $R(x(0),y(0)) = R(\bar x,\bar y) = 1$. Proposition \ref{prop:peak}(i) implies that $R$ increases along the backwards solutions. Together with \eqref{eq:symmetry}, this implies 
					\tcb{}
					\be\label{eq:R_back}
					R(x(t),y(t))>1\,, \quad \forall t > 0\,.
					\ee 
					Therefore, the trajectory enters the set of points with  reproduction number larger than $1$ for positive times $t$ and this set can be left only if the trajectory leaves  $\mc S$ (we remark that  $\mc S$ is not positively invariant under the backward dynamics). This in turn implies that $y(t)$ is strictly decreasing, while $x(t)$ is strictly increasing while in $\mc S$. 
					Now, note that either $(x(t), y(t))$ belongs to $\mc S$ for all positive times $t$, in which case the solution converges to a point in the boundary of $\mc S$ as $t\to +\infty$, or it hits the boundary of $\mc S$ at some finite time. Let us indicate the convergence or hit point by $(\hat x, \hat y)$. We note that, in the first case, such point must be an equilibrium, so that necessarily $\hat y=0$ from Proposition~\eqref{prop:general} (note that the equilibrium points of the backward dynamics coincide with the ones of the original dynamics). In the second case, instead, such a point $(\hat x, \hat y)$ must lie on the boundary of $\mc S$, with $\hat y$ in $(0, \bar y ]$ and $\hat x = 1-\hat y$,  so that $\hat x + \hat y = 1$. 
					In both cases, the strict monotonicity of $y(t)$ ensures that the first component of the orbit $\Gamma^-(\bar x,\bar y)$ can be represented as the graph of a function of $y$, denoted $\lambda(y):=x(y)$ and defined on $[\hat y,\bar y]$. 
					By construction, the function $\lambda$ is of class $\mc C^1$ since both $x(t)$ and $y(t)$ are $\mc C^1$ and $y(t)$ is strictly decreasing. Moreover, $\lambda$ is a strictly decreasing function because the strict monotonicity of $x(t)$ and $y(t)$ implies $dx /dy = \dot x/ \dot y < 0$. This proves item (i).  Item (ii) follows from the fact that $\hat y>0$ if and only if $(x(t),y(t))$ hits the boundary of $\mc S$, so that necessarily $\lambda(\hat y) + \hat y = 1$. 
					The fact that $\lambda(\bar y) = \kappa(\bar y) = \bar x$ follows by construction. The fact that $R(\lambda(y),y)>1$ for every $[\hat y,\bar y]$ follows from \eqref{eq:R_back}, proving item (iii). 
				\end{proof}
				
				The negative semi-orbit $\Gamma^-(\bar x,\bar y)$ is represented by the thick blue curve in Figure~\ref{fig:phase}. The green and the light-blue area are, respectively, the sets $\mc D_{\bar y}^-$ and $\mc D_{\bar y}^+$.  If the negative semi-orbit that starts from $(\bar x,\bar y)$ intercepts a point $(x,y)$ on the right boundary of $\mc S$, then $\hat y > 0$, otherwise $\hat y = 0$.
				
				Lemma~\ref{lemma:negorb} allows for the partition of $\mc D_{\bar y}$ in two subsets
				\be\label{D+}
				\ba{rcl}
				\mc D_{\bar y}^+ & = & \{(x,y)\in \mc D_{\bar y} : y\in [\hat y,\bar y], x>\lambda(y)\}\,, \\[6pt]
				\mc D_{\bar y}^- & = & \mc D_{\bar y} \setminus\mc D_{\bar y}^+.
				\ea
				\ee
				\begin{remark}\label{remark:D+}
					Note that, since $R(\lambda(y),y)>1$ for every $y$ in $[\hat y,\bar y)$ by Lemma \ref{lemma:lambda}(iii) and since $R(\bar x,\bar y) = 1$, $\mc D_{\bar y}^+$ contains all points with reproduction number larger than $1$ due to the monotonicity condition \eqref{condition3}. Therefore, this yields that $\rho(x,y)>0$ for all $(x,y)$ in $\mc D_{\bar y}^+$.
				\end{remark}
				\begin{remark}\label{remark:D}
					Let Assumption~\ref{ass:ass1} hold true. The analysis above holds under \eqref{Rmax>1} and $\bar y \leq \tilde y$. If \eqref{Rmax>1} does not hold, $\tilde y$ and $\hat y$ are not defined and $R(x,y)\leq 1$ for every $(x,y)$ in $\mc S$. Hence, the set $\mc D_{\bar y}^+$ is empty  by Remark \ref{remark:D+} and $\mc D_{\bar y} = \mc D_{\bar y}^-$.
					Instead if \eqref{Rmax>1} holds and $\bar y > \tilde y$, then $\bar x$ does not exist, therefore neither $\hat y$, the function $\lambda$ and the sets $\mc D_{\bar y}^+$ and $\mc D_{\bar y}^-$ are well defined. In both cases, the optimal control is trivial because the uncontrolled trajectory will stay in $\mc D_{\bar y}$.
				\end{remark}
				
				
				\subsection{Dynamic behavior of the uncontrolled BF-SIR epidemic model}
				In this part, we discuss the behavior of the uncontrolled BF-SIR epidemic model \eqref{uncontrol-system} with initial conditions either in $\mc D_{\bar y}^-$ or in $\mc D_{\bar y}^+$. In the first case, we show that the orbit remains inside $\mc D_{\bar y}^-$ and thus, in particular, the number of infected agents never exceeds $\bar y$. In this case, $u \equiv 0$ for all $t$, is the optimal control. In contrast, if the initial condition lies in $\mc D_{\bar y}^+$,  the trajectory will leave $\mc D_{\bar y}$ at some finite time, violating the threshold constraint. We also show that the abscissa of the hitting point is a regular function of the initial condition, a fact that will play a crucial role in the proof of Theorem \ref{theo:theo-control}.
				\begin{lemma}\label{lemma:D0} Let Assumption~\ref{ass:ass1} hold true. Then, the set $\mc D^-_{\bar y}$ is positively invariant for the uncontrolled BF-SIR epidemic model \eqref{uncontrol-system}.
				\end{lemma}
				\begin{proof} 
					Let $\mc I=\{(x, \bar y) \in \mc S\,:\, x\leq \bar x\}$. Since $\mc S$ is positively invariant by Proposition~\ref{prop:general} 
					a solution  $(x(t), y(t))$ can exit $\mc D_{\bar y}^-$ by crossing either $\Gamma^-(\bar x,\bar y)$ or $\mc I$. The first option is not possible for the uniqueness of the solution of the ODE. 
						Since $R(\bar x,\bar y) = 1$ by construction, it follows from \eqref{condition2} that $R(x,y) < 1$ and thus $\dot y \leq 0$ for every point $(x,y) \in \mc I$. Hence, the trajectory cannot leave $\mc D_{\bar y}$ by crossing $\mc I$.
						This concludes the proof.
					\end{proof}\smallskip
					
					
					\begin{lemma}\label{lemma:h-c2}
						Let $\beta$ be of class $\mc C^2$ and let Assumption~\ref{ass:ass1} hold true. Then: 
						\begin{enumerate}
							\item[ (i)]  for every $(x,y)$ in $\mc D_{\bar y}^+$, there exists a finite time $t$ such that $\phi(t,x,y)$ does not belong to $\mc D_{\bar y}$;
							\item[ (ii)] for every $(x,y)$ in $\mc D_{\bar y}^+$, there exists a $\mc C^2$ map $h:\mc D_{\bar y}^+ \to (\bar{x}, 1-\bar{y}]$ such that $(h(x,y), \bar{y})$ belongs to $\Gamma^+(x,y)$  and $h(x,y) \le x$, where equality holds true if and only if $y = \bar y$;
							\item[(iii)] $h$ is a constant of motion of the uncontrolled BF-SIR epidemic model \eqref{uncontrol-system} in $\mc D_{\bar y}^+$, that is, $h(x,y) = \xi$
							for every $(x,y)$ in $ \Gamma^-(\xi, \bar y)$ with $\xi$ in $(\bar{x}, 1- \bar{y}]$;
							\item[ (iv)] for every $(x_0, y_0)$ in  $\Gamma^-(\bar x,\bar y)$,
							\be\label{limit-h}\lim\limits_{\small\begin{array}{l}(x,y)\to (x_0, y_0)\\ (x,y)\in\mc D_{\bar y}^+\end{array}}h(x,y)=\bar x\,.\ee
						\end{enumerate}
					\end{lemma}
					\begin{remark}
						The quantity $h(x,y)$ is the abscissa of the crossing point of the orbit that starts from $(x,y)$ in $\mc D_{\bar y}^+$ with segment $\mc J$, i.e. the fraction of susceptible agents reached by the dynamics when the fraction of infected agents hits the threshold $\bar{y}$. An example of trajectory that starts in $\mc D_{\bar y}^+$ is illustrated in Figure~\ref{fig:phase} by the light-blue dashed curve.
					\end{remark}
						\begin{proofof}{Lemma \ref{lemma:h-c2}} (i) 
							Recall from Remark \ref{remark:D+} that $R(x,y)>1$ for every $(x,y)$ in $\mc D_{\bar y}^+$. This implies that every trajectory in $\mathcal{D}_{\bar y}^+$ is such that $y(t)$ is strictly increasing in $t$ as long as the solution remains in $\mc D_{\bar y}^+$. It thus follows from  Proposition~\ref{prop:peak}(ii) that the solution must necessarily exit $\mc D_{\bar y}^+$  in finite time. Since $\mc S$ is positively invariant by Proposition \ref{prop:general}, 
							the trajectory can leave $\mc D_{\bar y}^+$ by crossing either $\Gamma^-(\bar x,\bar y)$ or the set
							\be\label{J-def}\mc J=\{(x, \bar y) \in \mc S\,:\, \bar x<x\leq 1-\bar y\}\,.\ee
							Since $\Gamma^-(\bar x,\bar y)$ cannot be crossed because of the uniqueness of the solutions, the trajectory must necessarily cross $\mc J$ in finite time. This implies that the trajectory leaves $\mc D_{\bar y}$ in finite time.
							
							(ii) Given $(x_0,y_0)$ in $\mc D_{\bar y}^+$ and $t\geq 0$, let 
							$$
							\phi(t,x_0,y_0) = (x(t,x_0,y_0),y(t,x_0,y_0))
							$$
							be the solution of \eqref{uncontrol-system} at time $t$ with initial condition $(x_0,y_0)$. Now, note that, since the right hand side of \eqref{uncontrol-system} is of class $\mc C^2$, $\phi$ is of class $\mc C^2$ both in $t$ and in $(x_0, y_0)$ \cite[Chapter 5]{HartmanODE}.
							We now define $T(x_0, y_0)$ to be the minimum time $t$ for which $\phi(t,x_0,y_0)\in \mc J$, well defined because of  item (i),
							and define $h:\mc D_{\bar y}^+\to (\bar{x}, 1-\bar y]$ by putting
							$$h(x_0,y_0)=x(T(x_0, y_0),x_0,y_0)\,.$$
							Note that, for every $(x_0,y_0)$ in $\mc D_{\bar y}^+$, $y(T(x_0,y_0),x_0,y_0) =\bar y$ by construction and $\dot{y}(T(x_0, y_0),x_0,y_0)>0$ because of Remark \ref{remark:D+}. A direct application of the Implicit Function Theorem \cite[Theorem 9.18]{rudin1976principles} (on all points of $\mc D_{\bar y}^+$ including the boundary segment $\mc J$) together with the regularity of $\phi$ yields that $T$ is of class $\mc C^2$ in $(x_0, y_0)$. Finally, the $\mc C^2$ regularity of $h$ follows from the analogous regularity of $T$ and $\phi$. Note also that by Proposition \ref{prop:general}(i), $h(x_0,y_0)=x(T(x_0, y_0),x_0,y_0)\leq x_0$. 
							
							(iii) The fact that $h$ is a constant of motion of \eqref{uncontrol-system} is a consequence of its definition, since  the trajectories that start from points that belong to the same negative semi-orbit will hit segment $\mc J$ in the same point because of the uniqueness of the solutions.
							
							(iv) If $(x_0, y_0)=(\bar x, \bar y)$, the thesis is straightforward: fixed an arbitrary $\epsilon >0$, for every $(x,y)\in B_{\epsilon}(x_0, y_0)\cap\mc D_{\bar y}^+$ we have that, thanks to the definition of $h$ and item (ii),
							$$\bar x < h(x,y)  \leq x\leq \bar x+\epsilon$$
							This yields relation \eqref{limit-h}.
							
							Consider now $(x_0,y_0)$ in {$\Gamma^-(\bar x,\bar y)\setminus \{(\bar x, \bar y)\}$} and let $T(x_0, y_0)>0$ be the time such that $\phi(T(x_0, y_0), x_0, y_0)=(\bar x, \bar y)$. Fix any positive $\epsilon<||(\bar x, \bar y)-(x_0, y_0)||$ and let $T_{\epsilon} < T(x_0, y_0)$ be such that \be\label{eq:proof1}||\phi(T_{\epsilon}, x_0, y_0)-(\bar x, \bar y)||=\epsilon/2\ee Such a time $T_{\epsilon}$ exists because of the continuity of the solution. 
							Let $\epsilon_y:=\bar y-y(T_{\epsilon}, x_0, y_0)\in (0,\epsilon/2)$. By regularity with respect to the initial condition \cite[Chapter 5]{HartmanODE}, there exists $\delta >0$ such that, for every $(x_\delta,y_\delta)\in B_{\delta}(x_0, y_0)\cap\mc D_{\bar y}^+$,
							\be\label{eq:proof2}
							||\phi(t, x_0, y_0)-\phi(t, x_\delta, y_\delta)||<\epsilon_y \quad \forall t\in [0, T_{\epsilon}].
							\ee 
							This implies that $\phi(t, x_\delta, y_\delta)$ belongs to $\mc D_{\bar y}^+$ for all $t\in [0, T_{\epsilon}]$ since $\bar y - y(T_{\epsilon}, x_0, y_0) = \epsilon_y$ and $|y(T_{\epsilon}, x_0, y_0)- y(T_{\epsilon}, x_\delta, y_\delta)| < \epsilon_y$ imply $y(T_{\epsilon}, x_\delta, y_\delta)<\bar y$.
							As a consequence of  item (ii), we thus have  that \be\label{dis-invariant}h(\phi( T_\eps,x_\delta,y_\delta)) \le x(T_\eps,x_\delta,y_\delta)\,.\ee
							Hence,		$$
							\begin{aligned}
								\bar x < h(x_\delta,y_\delta) & = h(\phi(T_\eps,x_\delta,y_\delta)) \\
								& \le x( T_\eps,  x_\delta, y_\delta)  < \bar x+\epsilon/2+\epsilon/2 = \bar x+\epsilon\,,
							\end{aligned}$$
							where the first inequality follows from the definition of $h$, the first equality from item (iv), the second inequality from \eqref{dis-invariant} and the last inequality from \eqref{eq:proof1}-\eqref{eq:proof2}.
							Because of the arbitrariness of $\epsilon>0$, this yields the limit relation \eqref{limit-h}.

					\end{proofof}

					We now establish a technical result that characterizes the function $h$.
					\begin{lemma}\label{lemma:h}
						Let $\beta$ be of class $\mc C^2$ and let Assumption~\ref{ass:ass1} hold true. Then,  for every $(x,y)$ in $\mc D_{\bar y}^+$:
						\begin{enumerate}
							\item[(i)]
							\be\label{h2} h_x(x,y)R(x,y) = h_y(x,y)(R(x,y)- 1);
							\ee
							\item[(ii)]
							\begin{equation}\label{sufficient} 
								0\leq h_y(x,y) 
								\leq  \frac{1}{\rho(h(x,y), \bar y)}\,;
							\end{equation}
							\item[(iii)] if $y = \bar y$,
							$$
							h_y(x,y) = \frac{1}{\rho(h(x,y), \bar y)}\,.
							$$
						\end{enumerate}
					\end{lemma}\medskip
					\begin{proof} 
						(i) It follows from Lemma  \ref{lemma:h-c2}(iii) that
						$\dot h = h_x\dot{x} + h_y \dot{y} = 0  
						$
						for every $(x,y)$ in $\mc D_{\bar y}^+$  under the uncontrolled BF-SIR model \eqref{uncontrol-system}.
						Using \eqref{uncontrol-system}, we then get item (i).
						
						(ii)-(iii) We now show that $h_x(x,y) \ge 0$ on $\mc D_{\bar y}^+$. For $(x_0,y_0)$ in $\mc D_{\bar y}^+$, choose $\epsilon >0$ such that $(x_0-\epsilon,y_0)$ belongs to $\mc D_{\bar y}^+$. Denote by $(x(t),y(t))$ and $(x_{\epsilon}(t), y_{\epsilon}(t))$ the solutions associated with initial conditions $(x_0,y_0)$ and $(x_0-\epsilon,y_0)$, respectively. As long as they remain in $\mc D_{\bar y}^+$, the second components $y(t)$ and $y_{\epsilon}(t)$ are both increasing in time, so that the corresponding orbits can be described as graphicals of functions, denoted respectively $x=f(y)$ and $x=f_{\epsilon}(y)$ with $y$ in $[y_0, \bar y]$. Since the orbits cannot intersect and since $x_0=f(y_0)>x_0-\epsilon =f_{\epsilon}(y_0)$, it must be $h(x_0, y_0)=f(\bar y)>f_{\epsilon}(\bar y)=h(x_0-\epsilon, y_0)$. Given the arbitrariness of $\epsilon$, this implies that $h_x(x_0,y_0)\geq 0$. Moreover, since $x\mapsto h(x, \bar y)$ is the identity function on $(\bar x, 1]$, we have that $h_x(x, \bar y)=1$. Consequently, thanks to \eqref{h2} and thanks to the fact that $R(x,y)>1$ for every $(x,y)$ in $\mc D_{\bar y}^+$ because of  Remark \ref{remark:D+}, we also have that 
						\begin{align}
							h_y(x,y)&\geq 0\;\quad\forall (x,y)\in\mc D_{\bar y}^+, \label{dh-dy} \\[6pt]
							h_y(x,\bar y)&=\frac{R(x,\bar y)}{R(x,\bar y)-1}=\frac{1}{\rho(x,\bar y)}\,. \label{dh-dy2}
						\end{align}
						Since $h(x,y) = x$ whenever $y = \bar y$, \eqref{dh-dy2} proves item (iii), while \eqref{dh-dy} proves the left hand inequality of item (ii). To conclude the proof of item (ii), note that
						\be\begin{array}{rcl}\label{dt-h}
							\dot h_y & = &\dot{x} \, h_{xy}+ \dot y \, h_{yy} \\[5pt]
							& = &-R\gamma y h_{xy}+(R-1)\gamma y h_{yy}\,.\end{array}\ee
						Taking the partial derivative of \eqref{h2} with respect to $y$ and then substituting \eqref{h2} inside this derivative, we get
						$$\begin{array}{rcl}h_{xy}
							R&=&h_{yy}
							(R-1)+\left(h_y-h_x
							\right)R_y
							\\[5pt]
							&=&h_{yy}
							(R-1)+R^{-1}h_yR_y\,.
						\end{array}
						$$
						Plugging this expression into \eqref{dt-h} we finally get
						\be\label{dt-h2} 
						\dot{h}_y
						=-\gamma yR^{-1}h_yR_y
						\geq 0\,,\ee
						whose sign follows from \eqref{condition3} and \eqref{dh-dy}.
						Since $h_y$ is non-decreasing in time by \eqref{dt-h2} and since $(h(x, y),\bar y)$ belongs to the positive orbit $\Gamma(x,y)$, we get that
						$$h_y(x,y)\leq h_y(h(x,y),\bar y)\,.
						$$
						for every $(x,y)$ in $\mc D_{\bar y}^+$. Combining this with \eqref{dh-dy2} we obtain the thesis. 
					\end{proof}
					
					\begin{remark}
						Lemma~\ref{lemma:h-c2} and Lemma~\ref{lemma:h} guarantee the existence of $h(x,y)$ for every $(x,y)$ in $\mc D_{\bar y}^+$ under conditions \eqref{condition2}-\eqref{condition3} and characterize $h(x,y)$. Although its existence and regularity is essential to prove our main result, the explicit form of $h(x,y)$ cannot be found unless in some special cases. For instance, the explicit form of $h(x,y)$ may be found for the classical SIR epidemic model and for the BF-SIR epidemic model with feedback mechanism of Example \ref{ex2}, due to the fact that the BF-SIR epidemic model with such feedback mechanism and the classical SIR model admit a non-trivial invariant of motion. We refer to \cite{CDC2025}, \cite{Miclo.ea:2022} for more details.
					\end{remark}

					\section{Proof of optimality}\label{sec:proof-optimality}
					In this section we provide the proof of Theorem~\ref{theo:theo-control}. This is obtained by first giving a characterization of the candidate value function in Section~\ref{sec:value} (i.e., the cost of the candidate optimal control) and then proving optimality in Section~\ref{sec:proof}.
					\subsection{Candidate value function}\label{sec:value}
					We start by computing the cost of the candidate optimal control signal $u^*(t)$ defined in \eqref{feedback}--\eqref{eq:mu}. 
					
					As pointed out in Remark \ref{remark:D}, if \eqref{Rmax>1} is not verified, or if \eqref{Rmax>1} holds true and $\tilde y \le \bar y$ (cf. Lemma \ref{lemma:kappa} for the definition of $\tilde y$), then $\mc D_{\bar y}^- = \mc D_{\bar y}$, so that the optimal control is $u^* \equiv 0$ by Lemma \ref{lemma:D0}. Hence, in the next part we shall work under the assumption that \eqref{Rmax>1} holds true and $\tilde y > \bar y$. We let $T_0=T(x_0, y_0)$ be the first time the solution of \eqref{uncontrol-system} with initial condition $(x_0, y_0)$ in $\mc D^+_{\bar y}$ hits the threshold $\bar y$ and let  $$T_1=T_0+\frac{h(x_0, y_0)-\bar x}{\gamma \bar y}\,,$$
					where the function $h$ was introduced in Lemma~\ref{lemma:h-c2}.
					

						\begin{proposition}\label{lemma:Ju*}
							Let Assumption~\ref{ass:ass1} hold true and assume that a threshold $\bar y$ and an initial condition $(x_0, y_0)$ in $\mc D^+_{\bar y}$ have been fixed. 
							Consider the control function
							\be\label{def-control} u^*(t)=\left\{\begin{array}{ll} 0\quad &\hbox{if}\; t< T_0\;\hbox{or}\; t> T_1\\[5pt]
								\rho( \tilde x(t), \bar y)\quad &\hbox{if}\; T_0\leq t\leq T_1
							\end{array}\right.\ee
							where $$ \tilde x(t)=h(x_0,y_0) - \gamma(t-T_0)\bar y\,.$$
							The corresponding solution of the CBF-SIR epidemic model \eqref{control-system} $(x^*(t),y^*(t))$ has the following properties:
							\begin{enumerate}
								\item[(i)] $(x^*(t),y^*(t))$ admits the following form for $T_0<t<T_1$:
								\be\label{eq:solution_control}
								\begin{cases}
									x^*(t) = h(x_0,y_0) - \gamma(t-T_0)\bar y, \\ 
									y^*(t) = \bar y. \\
								\end{cases}
								\ee
								\item[(ii)] The feedback relation $u^*(t) = \mu(x^*(t), y^*(t))$ holds for every $t\geq 0$.
								\item[(iii)] $u^*$ is feasible, namely $y^*(t)\leq\bar y$ for every $t\geq 0$.
								\item[(iv)]	The cost of $u^*$ is given by	$$\label{Ju*} J(u^*) =\displaystyle \frac1{\gamma \bar{y}}\int_{\bar{x}}^{h(x_0,y_0)}   \rho(x,\bar y) \de x\,.	$$
							\end{enumerate}
						\end{proposition}
						
					
					\begin{proof}
						(i) By construction, up to time $T_0$ the evolution of the CBF-SIR epidemic model \eqref{control-system} coincides with the uncontrolled case. By the way $T_0$ is defined, it then follows that
						$$(x^*(T_0), y^*(T_0))=(h(x_0, y_0), \bar y).$$
						We observe that the vector function in \eqref{eq:solution_control} coincides with $(x^*(T_0), y^*(T_0))$ for $t=T_0$. Moreover a direct check shows that such a vector function is a solution of \eqref{control-system} with control $u^*$ in $[T_0, T_1]$. This proves item (i). Item (ii) is a direct consequence of the definition of $\mu$ and of item (i). To prove item (iii) notice that, by construction, the inequality $y^*(t)\leq\bar y$ is satisfied till time $T_1$ and the solution $(x^*(t), y^*(t))$ for $t>T_1$ coincides with the solution of the uncontrolled model with initial condition $x^*(T_1)=\bar x $ and $y^*(T_1)=\bar y$. Since $R(\bar x, \bar y)=1$, it follows from Proposition~\ref{prop:peak}(i) that $y^*(t)$ is strictly decreasing from this time on. This yields $y^*(t)  <\bar y$ for $t>T_1$ and proves item (iii).	
						Finally,
						\begin{align*}
							J(u^*)  &=
							\int_{0}^{+\infty} u^*(t) \de t = \int_{T_0}^{T_1}u^*(t)\de t \\ & = \int_{T_0}^{T_1}\rho(x^*(t),y^*(t))\de t \
							= \ \int_{T_0}^{T_1}\rho(x^*(t),\bar y)\de t \\
							&= \ \frac1{\gamma \bar{y}}\int_{\bar{x}}^{h(x_0,y_0)}\rho(x,\bar y)\de x\,,
						\end{align*}
						where the last two equivalences follow from item (i) by a change of variables.
						This proves item (iv) and concludes the proof.
					\end{proof}\smallskip
					
					Motivated by Proposition~\ref{lemma:Ju*}, we define our candidate value function $V:\mc D_{\bar y}\to\R$ by
					\be\label{value-function} V(x,y) =  \begin{cases}
						0 & \text{if} \ (x,y) \in \mc D_{\bar y}^- \\
						\displaystyle  \frac{1}{\gamma \bar{y}}	\int_{\bar{x}}^{h(x,y)} \rho(s,\bar y) \de s	 & \text{if} \ (x,y) \in \mc D_{\bar y}^+.
					\end{cases} 
					\ee
					
					The following result provides some properties of the candidate value function that are instrumental for the proof of Theorem~\ref{theo:theo-control}.
					\begin{lemma}\label{lemma:V} Assume that $\beta$ is a $\mc C^2$ function and let Assumption~\ref{ass:ass1} hold true. Consider $V$ defined in \eqref{value-function}, then:
						\begin{enumerate}
							\item[(i)]$V$ is continuous on $\mc D_{\bar y}$ and $\mc C^1$ on $\mc D_{\bar y}^+$;
						\item[(ii)] for every $(x,y)$ in $\mc D_{\bar y}^+$,
						\begin{align}
							\ds V_y(x,y) &=\frac{\rho(h(x,y),\bar y) }{\gamma \bar{y}}h_y(x,y) , \label{eq:V_y}\\[1ex]
							\ds V_x(x,y) &= \rho(x, y) V_y(x,y)\,; \label{eq:V_x} \end{align}
						\item[(iii)] $0 \leq V_y(x,y) \leq 1/(\gamma \bar y)$ for every $(x,y)$ in $\mc D_{\bar y}^{+}$;
						\item[(iv)] $V_y(x,y) = 1/(\gamma \bar y)$ for every $(x,y)$ in $\mc D_{\bar y}^+$ with $y = \bar y$.
					\end{enumerate}
				\end{lemma}
				\begin{proof}
					(i)  Note that $V(x,y) = 0$ for every $(x,y)$ in $\Gamma^{-}(\bar x,\bar y)$ since $\Gamma^{-}(\bar x,\bar y)$ belongs to $\mc D_{\bar y}^-$. Moreover, it follows from Lemma \ref{lemma:h-c2}(iv) and from the definition of $V$ that for every $(x,y)$ in $\Gamma^-(\bar x,\bar y)$
					$$\lim\limits_{\small\begin{array}{l}(x_0,y_0)\to(x, y)\\ (x_0,y_0)\in\mc D_{\bar y}^+\end{array}}V(x_0,y_0)=0\,.$$
					This, together with the continuity of $h(x,y)$ established in Lemma~\ref{lemma:h-c2}(ii), proves continuity of $V$ on $\mc D_{\bar y}$. The differentiability of $V$ on $\mc D_{\bar y}^+$ follows from the differentiability of $h$ established in Lemma~\ref{lemma:h-c2}(ii). 
					
					(ii) The statement follows from \eqref{value-function} and \eqref{h2} by direct computation.
					
					(iii) The fact that $V_y(x,y) \geq 0$ follows from item (ii), from Lemma \ref{lemma:h}(ii) and from the fact that $\rho(x,y) \ge 0$ for every $(x,y)$ in $\mc D_{\bar y}^+$ \tcm{by Remark \ref{remark:D+}}. The right-hand inequality follows immediately from Lemma \ref{lemma:h}(ii).
					
					(iv) The statement is a straightforward consequence of Lemma \ref{lemma:h}(iii). 
				\end{proof}
				
				\begin{remark}\label{remark:v-const}
					Note that since $h$ is a constant of motion of the uncontrolled BF-SIR epidemic model \eqref{uncontrol-system} in $\mc D_{\bar y}^+$  (cf. Lemma \ref{lemma:h-c2}(iv)), the same property also holds for $V$.
				\end{remark}
				
				Lemma~\ref{lemma:V} states that the candidate value function $V(x,y)$ is continuous on $\mc D_{\bar y}$, but continuously differentiable on $\mc D_{\bar y}^+$ only. In fact, it is not hard to prove that $V(x,y)$ is not continuously differentiable on $\Gamma^-(\bar x,\bar y)$. 
				For this reason, we cannot apply the Hamilton-Jacobi-Bellman equation \cite{liberzon, SCC.Kamalapurkar.Walters.ea2018} to prove that $u^*$ is the optimal control. In the next section we shall provide an alternative proof for our main result.

				\subsection{Proof of Theorem \ref{theo:theo-control}}\label{sec:proof}
				The next result proves that, for every feasible control with finite cost, the CBF-SIR model \eqref{control-system} enters the region with reproduction number less than one in finite time.
				\begin{lemma}\label{lemma:R<1}Let $u$ in $\mc U_{x_0, y_0}^{\bar y}$ be such that $J(u)<+\infty$, and let $(x(t), y(t))$ be the corresponding solution of the CBF-SIR epidemic model \eqref{control-system} with initial condition $(x_0,y_0)$ in $\mc D_{\bar y}$. Then, there exists a finite time $\bar t$ such that $R(x(\bar t), y(\bar t))<1$. 
				\end{lemma}
				\begin{proof} 
					%
					%
					Assume by contradiction that $R(x(t), y(t))\geq 1$ for every time $t$ and let 
					$\hat R$ be the maximum of the function $R$ over the compact set $\mc S$. Integrating the second equation in \eqref{control-system} and using the inequality
					$$(1-u(t))R(x(t),y(t))-1\geq -u(t)R(x(t),y(t))\geq -u(t)\hat R$$
					yields
					\begin{align*}
						y(t) &=y_0 \exp\left(\gamma \int_{0}^{t} \left( (1-u(\tau))R(x(\tau), y(\tau)) - 1\right)\de\tau \right)\\
						&\geq y_0 \exp\left(\gamma \hat R \int_{0}^{t} \left( -u(\tau) \right)\de\tau \right)\,.
					\end{align*}
					Since $\lim_{t \to +\infty} y(t) = 0$ by Proposition \ref{prop:general}(iv), we get that 
					$$J(u)=\lim\limits_{t \to +\infty}  \int_{0}^{t}  u(\tau) d\tau=+\infty\,,$$
					proving the statement.
				\end{proof}
				
				We can now proceed to prove Theorem \ref{theo:theo-control}.
				
				If $(x_0,y_0)$ belongs to $\mc D_{\bar y}^-$, the uncontrolled solution remains below $\bar y$ at all times thanks to Lemma~\ref{lemma:D0}. Therefore $u \equiv 0$ is the unique optimal control that satisfies Theorem~\ref{theo:theo-control}.
				
				We now assume that $(x_0,y_0)$ belongs to $\mc D_{\bar y}^+$ and we show that the control signal defined in \eqref{def-control} is optimal. 
				
				To this end, consider a feasible control $u$ in $\mc U_{x_0,y_0}^{\bar{y}}$ such that $J(u)<+\infty$ and let $(x^u(t),y^u(t))$, for $t \ge 0$, be the solution of \eqref{control-system} with such control signal and initial state $(x_0,y_0)$. 
				Since $R(x,y)>1$ for every $(x,y)$ in $\mc D_{\bar y}^+$ by Remark \ref{remark:D+}, the solution $(x^u(t),y^u(t))$ must exit $\mc D_{\bar y}^+$ at some finite time because of Lemma~\ref{lemma:R<1} and because $J(u)<+\infty$. By the way $\mc D_{\bar y}^+$ is defined (equations \eqref{D+}), it follows that the boundary of this subset in the positively invariant set $\mc S$ consists of the graphical of $\lambda$ and of interval $\mc J$ defined in \eqref{J-def}. As $u$ is assumed to be feasible, necessarily the exit point must be located in the graphical of $\lambda$. In other words, there exists a finite time $t^\circ$ such that
				\begin{align}
					y^u(t^\circ)\leq \bar y,\quad &x^u(t^\circ)=\lambda(y^u(t^\circ))\,, \label{exit-point}  \\[1ex]
					(x^u(t),y^u(t))&\in \mc D_{\bar y}^+\quad\forall t<t^\circ\,. \label{exit-point2}
				\end{align}
				Let $u^\circ $ be the control signal defined as
				\be\label{cut-}u^\circ(t)=\left\{\begin{array}{ll} u(t)\; &t\in [0, t^\circ)\\ 0 &t\in [t^\circ, +\infty)\,,\end{array}\right.\ee
				and let $(x^\circ(t),y^\circ(t))$, for $t\ge0$, be the solution of \eqref{control-system} with control signal $u^\circ$ and initial state $(x_0,y_0)$. Since $u$ is piecewise-$\mc C^1$, also $u^\circ$, by construction, is piecewise-$\mc C^1$. The solution $(x^\circ(t),y^\circ(t))$ coincides with $(x^u(t),y^u(t))$ on $[0, t^\circ]$ while coincides with a trajectory of the uncontrolled model for $t>t^\circ$ with an initial condition $(x^u(t^\circ), y^u(t^\circ))$.  In particular, since $(x^\circ(t^\circ),y^\circ(t^\circ))$ belongs to $\mc D_{\bar y}^-$, Lemma \ref{lemma:D0} guarantees that $(x^\circ(t),y^\circ(t))$ belongs to $\mc D_{\bar y}^-$ for all times $t \ge t^\circ$. This yields $y^\circ(t)\leq\bar y$ for all times $t$. In other terms, the control $u^\circ$ is feasible. Moreover, 
				\be\label{control-init}
				\begin{aligned}	
					J(u^\circ) & = \int_{0}^{+\infty} u^{\circ}(t) \, \de t   \ = \int_0^{t^\circ} u^\circ(t) \de t \\
					& = \int_{0}^{t^\circ} u(t) \, \de t \leq \! \int_{0}^{+\infty} u(t) \, \de t = J(u).
				\end{aligned}
				\ee
				
				In the following, we compare the candidate optimal control $u^*$ with the control $u^\circ$. In order to do this, we need to preliminarily study the regularity of the composed function $V(x^\circ(t), y^\circ(t))$. Because of property \eqref{exit-point2} and Lemma~\ref{lemma:V}(i), we have that $V(x^\circ(t), y^\circ(t))$ is piecewise-$\mc C^1$ on $[0, t^\circ)$ and, trivially, on $(t^\circ, +\infty)$. Moreover, for every $t<t^\circ$ where the derivative exists, it satisfies
				\be \label{dotV}\begin{aligned} \!\!\! \dot V(x^\circ(t),y^\circ(t))
					&= V_x(x^\circ,y^\circ)\dot{x}^\circ + V_y(x^\circ,y^\circ)\dot{y}^\circ\\[3pt]
					&=  \big(\rho(x^\circ,y^\circ)\dot{x}^\circ + \dot{y}^\circ\big)V_y(x^\circ,y^\circ)\\[3pt]
						& = -  \gamma y^\circ(t) u^\circ(t) V_y(x^\circ(t),y^\circ(t))  \\[3pt]
						& \ge -  \frac{y^\circ(t)}{\bar y} u^\circ(t)\\[3pt]
						& \ge -  u^\circ(t)
					\end{aligned}\ee
					where the second equality follows from \eqref{eq:V_x}, 
					the first inequality from the right-hand inequality of Lemma~\ref{lemma:V}(iii) and the last one from the feasibility of $u^{\circ}$. On the other hand, \eqref{value-function} implies that 
					\be\label{eq:dotV0D-}
					\dot V(x^\circ(t),y^\circ(t)) = 0 = u^\circ(t)\,, \quad \forall t > t^\circ\,,
					\ee since $(x^\circ(t),y^\circ(t))$ belongs to $\mc D_{\bar y}^{-}$ for such times. 
					
					

					%
					%
					We now derive an upper bound for $\dot V(x^\circ(t),y^\circ(t))$ for every $t$ in $[0, t^\circ)$ where the derivative exists by
					\be \label{dotV+}\begin{aligned} \!\!\! \! \dot V(x^\circ(t),y^\circ(t))
						\!& \le \dot y^\circ V_y(x^\circ,y^\circ) \\[3pt]
						& =   \gamma y^\circ [R(x^\circ, y^\circ)(1-u^{\circ})-1]V_y(x^\circ,y^\circ) \\[3pt]
						& \le  \gamma y^\circ R(x^\circ, y^\circ) V_y(x^\circ,y^\circ)\\[3pt]
						& \le R(x^\circ(t), y^\circ(t))\\[3pt]
						& \le R(1, 0)
					\end{aligned}\ee
					where the first inequality follows from the monotonicity of $x^\circ(t)$ and the non-negativity of $V_x$, the equality from \eqref{control-system}, the second inequality from the left-hand inequality of Lemma \ref{lemma:V}(iii) and from 
					$R(x^\circ, y^\circ)(1-u^{\circ})-1 \le R(x^\circ, y^\circ)$, the third one from the right-hand inequality of Lemma \ref{lemma:V}(iii) and the last one from Assumption \ref{ass:ass1}. Inequalities \eqref{dotV} and \eqref{dotV+} together with relation \eqref{eq:dotV0D-} yield that $\dot V(x^\circ(t),y^\circ(t))$ is a bounded function on the complementary of a negligible set. Consequently, $V(x^\circ(t),y^\circ(t))$ is an absolutely continuous function on $[0, +\infty)$ and the fundamental theorem of calculus holds. From relations \eqref{dotV} and \eqref{eq:dotV0D-} we thus obtain
					\be\begin{aligned}\label{int1}  J(u^{\circ})& =\int_0^{+\infty}  
						u^\circ(s)\de s = \int_0^{t^\circ}  u^\circ(s) \de s \\[3pt]
						& \geq V(x^\circ(0),y^\circ(0)) -V(x^\circ(t^\circ),y^\circ(t^\circ))\\[5pt]
						&= V(x_0, y_0)\,,\end{aligned}\ee
					since $(x^\circ(t^\circ),y^\circ(t^\circ))$ belongs to $\mc D_{\bar y}^{-}$  and thus $V(x^\circ(t^\circ),y^\circ(t^\circ)) = 0$ by construction.

						Consider now the solution of the candidate optimal control $u^*$ and the associated solution $(x^*(t), y^*(t))$ with initial condition $(x_0, y_0)$. From Proposition~\ref{lemma:Ju*} (i), we know that the solution belongs to $\mc J \subseteq \mc D_{\bar y}^+$ for $t$ in $[T_0, T_1)$. For $t$ in $[T_0, T_1)$, we have the following chain of equalities:
						\be\label{eq:dotV*}
						\begin{aligned} \dot V(x^*(t), y^*(t))
							& = V_x(x^*,y^*)\dot{x}^* + V_y(x^*,y^*)\dot{y}^*\\[3pt]
							& =   V_x(x^*,\bar y)\dot{x}^* \\[3pt]
							& = -  \rho(x^*,\bar y) V_y(x^*,\bar y) \gamma\bar y \\[3pt]
							& =-  u^*(t)\,,
						\end{aligned}
						\ee
						where the  second equivalence follows from Proposition \ref{lemma:Ju*}(i) and from \eqref{eq:solution_control}, the third one from Proposition \ref{lemma:Ju*}(i), \eqref{eq:V_x} and from the fact that  $h(x^*(t),y^*(t)) = x^*(t)$  whenever $y^*(t) = \bar y$ (cf. Lemma \ref{lemma:h-c2}(ii)), and the last one from Lemma~\ref{lemma:V}(iv) and from \eqref{def-control}. Thanks to 
						the fact that $(x^*(T_0), y^*(T_0))\in \mc D_{\bar y}^+$, $(x^*(T_1), y^*(T_1))\in \mc D_{\bar y}^-$ and \eqref{eq:dotV*}, this yields
						\be\begin{aligned}\label{int2} 
							J(u^*) & = \int_0^{+\infty}  u^*(s)\de s = \int_{T_0}^{T_1}  u^*(s) \de s \\[3pt]
							&=  V(x^*(T_0), y^*(T_0)) - V(x^*(T_1), y^*(T_1))\\[3pt]
							& =  V(x^*(T_0), y^*(T_0)) \\[3pt]
							& = V(x_0, y_0)\\[3pt]
							& \le  J(u^{\circ})\,,
						\end{aligned}\ee
						where the last equality follows from Remark~\ref{remark:v-const} and the last inequality from \eqref{int1} since $u^*=0$ for all $t \in [0, T_0]$. 
						Finally, relations \eqref{control-init} and \eqref{int2} imply
						$$
						J(u) \ge J(u^{\circ}) \ge J(u^*)\,.
						$$
						As $u$ was chosen to be any feasible control function, this proves the optimality of $u^*$.  Since $u^*$ is feasible by Proposition~\ref{lemma:Ju*}(iii), the proof of existence in Theorem~\ref{theo:theo-control} is completed.
						
						We next prove uniqueness. Suppose $u$ in $\mc U_{x_0,y_0}^{\bar{y}}$ is optimal. Following previous considerations, similarly defining $t^o$ as the first time the corresponding solution $(x^u, y^u)$ exits the region $\mc D_{\bar y}^+$ and defining $u^o$ as in \eqref{cut-}, we notice that inequality \eqref{control-init} must now be an equality. This forces $u(t)=0$ for $t>t^o$. We now consider $u(t)$ for $t\leq t^o$ and we prove that $u$ must be $0$ before the solution reaches the threshold value $\bar y$. We define $\mc C=\{t\in [0, t^o]\,|\, y^u(t)<\bar y, u(t)>0\}$ and prove that, indicating with $\mu$ the Lebesgue measure, $\mu(\mc C)=0$. Suppose by contradiction that $\mu(\mc C)>0$. Then standard considerations on the Lebesgue integral imply that
						$$\int_{ \mc C} u(s) \de s>\int_{\mc C}  \frac{y^u(s)}{\bar y}u^\circ(s) \de s.$$
						Combining this with the second inequality in \eqref{dotV} and following steps as in \eqref{int1}, we obtain
						\be\label{int1bis} J(u)\! = \! \int_0^{t^\circ}  u(s) \de s 
						\!> \!
						\int_0^{ t^o}  \frac{y^u(s)}{\bar y}u^\circ(s) \de s\geq V(x_0, y_0) \!=\! J(u^*).
						\ee
						This contradicts the supposed optimality of $u$ and implies that $\mu(\mc C)=0$. Since $u$ is assumed to be right-continuous and $y$ is continuous, this implies that $u(t)=0$ for every $t$ such that $y^u(t)<\bar y$. In conclusion, $u$ has to be null as long as the second component $y^u(t)$ of the solution stays below the threshold value $\bar y$ and from the first time the solution exits $\mc D_{\bar y}^+$. Therefore, it can only be non null when the solution is such that $y^u(t)=\bar y$ and $x^u(t)>\bar x$. We now conclude by showing that, necessarily $u(t)=\rho(x^u(t), \bar y)$ in all such points. We first notice that we must have $u(t)\geq \rho(x^u(t), \bar y)$.  Indeed, if $u(\bar t)< \rho(x^u(\bar t), \bar y)$ for some $\bar t$, equations \eqref{control-system} and the right-continuity of $u$ imply that $u$ would not be feasible as it would allow the component $y^u$ to exceed $\bar y$ in a right neighborhood of $\bar t$. Moreover, if $u(\bar t)> \rho(x^u(\bar t), \bar y)$ at some point $\bar t$, then, similarly, $y^u$ would be decreasing in a right neighborhood of $\bar t$ and thus $u$ would be non null in instants where $y^u$ is below $\bar y$. As this fact was already shown to lead to non optimality, it is also forbidden for $u$. This proves that $u(t)=\rho(x^u(t), \bar y)$ whenever the solution is such that $y^u(t)=\bar y$ and $x^u(t)>\bar x$, and thus is necessarily identical to $u^*$.

						\section{Counterexample}\label{sec:counterexample}
						In this section we provide a numerical counterexample that highlights the role of Assumption~\ref{ass:ass1} in guaranteeing the optimality of the filling-the-box strategy. Specifically, we consider a transmission rate of the form $\beta(x,y) = (1 - 0.7x) y$, which does not satisfy Assumption~\ref{ass:ass1} and under this setting, we show that the filling-the-box strategy is no longer optimal.
						
						\begin{figure}
							\centering
							\subfloat[][]{\includegraphics[width=5cm, height=8cm]{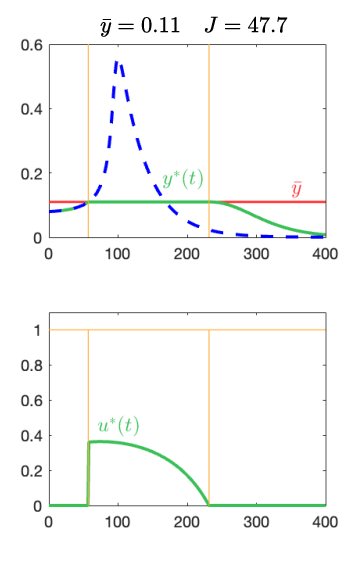}}
							\subfloat[][]{\includegraphics[width=5cm, height=8cm]{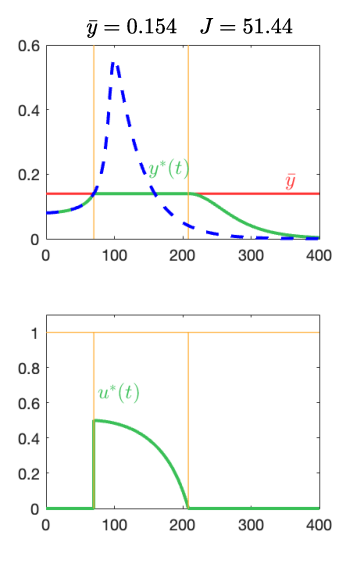}}
							\caption{Numerical counterexample illustrating the  filling-the-box strategy for two different thresholds: (a) $\bar y = 0.11$ and (b) $\bar y = 0.154$. In each case, the top plot shows the infection dynamics: the uncontrolled epidemic is represented by the blue dashed curve, while the controlled trajectory under the filling-the-box strategy is plotted in green. The bottom plot reports the corresponding control function applied over time. 
							}
							\label{fig:counterexample}
						\end{figure}
						
						To illustrate this, we compare  the cost of the filling-the-box strategy  for two different thresholds $\bar y$. The initial condition is fixed at $y_0 = 0.08$ and $x_0 = 1-y_0$. Figure~\ref{fig:counterexample} shows the two simulations. 
						
						When the threshold is set to $\bar y = 0.11$, the policy produces a total cost of $47.7$. Surprisingly, increasing the threshold to $\bar y = 0.154$ leads to a higher total cost of $51.44$. As the control signal implemented for $\bar y = 0.11$ is also feasible for $\bar y = 0.154$, this proves that the filling-the-box strategy is not optimal for this transmission rate. The explanation for this fact is that, since the reproduction number in this example is increasing in $y$, keeping the infection curve at the highest admissible level does still lead to shorter infection (i.e. one can deactivate the control sooner), but it requires a stronger intervention (as the bottom plots of Figure \ref{fig:counterexample} illustrate), resulting in a larger cumulated cost.
						This example shows that when Assumption~\ref{ass:ass1} is not satisfied, the simple filling-the-box strategy may fail to be optimal, so that no further extension of our optimality result is possible in general. At the same time, we have identified a broad and realistic class of infection-rate functions for which Assumption~\ref{ass:ass1} is satisfied, implying that the filling-the-box policy remains optimal in many practically relevant scenarios, while outside this regime the timing of interventions becomes critical and optimal policies may require more sophisticated strategies than a fixed-threshold rule.

						\section{Conclusion}\label{sec:conclusion}
						This paper has addressed the optimal control of a behavioral-feedback SIR epidemic model in which the infection rate depends in feedback on both the fraction of susceptible and infected agents. Under a general monotonicity condition on the feedback mechanism, the infection curve is proven to be unimodal as in the classical SIR model. Building on this property, we studied an optimal control problem that minimizes intervention costs while keeping the infection curve below a prescribed threshold over an infinite time horizon. Through a geometric analysis of the uncontrolled dynamics, we showed that, under mild regularity and monotonicity assumptions, the filling-the-box strategy is the unique optimal control, among the feasible right-continuous  control signals.  
						
						Our work suggests several directions for future research.  A first step would be to generalize the cost functional, for instance by introducing nonlinear intervention costs or by explicitly incorporating the social and economic burden of infections. Another natural extension is to consider delays in the infection rate, reflecting the realistic lag in observing and reacting to epidemic dynamics. Finally, extending the framework to structured populations or complex networks would provide further insights into the interplay between behavioral feedback and optimal control.

						\appendix
						\section{Proof of Proposition \ref{prop:general}}\label{sec:proof-wellposedness}
						The right-hand side of \eqref{control-system-full} is locally Lipschitz with respect to $(x,y,z)$. Hence, on every time interval on which $u$ is continuous, local existence and uniqueness of the solution for the corresponding Cauchy problem is a consequence of the Picard-Lindelöf theorem \cite[I.3]{hale}). Consider now any interval $\mc I$ admitting minimum $t_0$ on which $u$ is continuous. By integrating equations in \eqref{control-system-full}, assuming, for $t=t_0$, initial condition $(x_0, y_0,z_0)$ in $\Delta$, we get 
						\begin{align}
							x(t) &= x_0 \exp\left( \int_{t_0}^{t} \left( -(1-u(\tau))\beta(x(\tau), y(\tau))  y(\tau) \right)\de \tau \right), \label{eq:sol_x}\\
							y(t) &= y_0 \exp\left( \int_{t_0}^{t} \left( (1-u(\tau))\beta(x(\tau), y(\tau))  x(\tau)- \gamma \right)\de \tau \right), \label{eq:sol_y} \\
							z(t) &= z_0 + \gamma \int_{t_0}^{t} y(\tau) \, d\tau\,.
						\end{align}
						Moreover, the sum of the equations in \eqref{control-system-full} gives $\dot{x}(t) + \dot{y}(t) + \dot{z}(t)= 0$. This implies that, as long as the solution $(x(t), y(t),z(t))$ exists, we must have that \be\label{X}x(t) \geq 0, \, y(t) \geq 0, \, z(t) \geq 0, \quad x(t) + y(t)+z(t) = 1\,.\ee In other terms, as long as the solution exists, it lives in $\Delta$.  In particular, this implies that the solution is globally defined on the whole of $\mc I$ \cite[Section 3.3]{knauf}. We now construct a global solution on $[0,+\infty)$ in a recursive way. We order the discontinuity points as $t_1<t_2<\cdots$.
						First, we let $u^k:[t_k, t_{k+1}]\to\R$ for $k\geq 0$ (letting $t_0=0$) to be the unique continuous function that coincides with $u$ on $[t_k, t_{k+1}]$. Similarly, we define $u^k:[t_k, +\infty)\to\R$ when the set of discontinuities is finite and has cardinality $k$. We now define $(x(t),y(t),z(t))$ on $[t_0, t_1]$ as the unique solution of the Cauchy problem relative to \eqref{control-system-full} with $u=u_0$ and initial condition $(x_0,y_0,z_0)$. Assuming the solution $(x,y,z)$ to be constructed up to $t_k$ with the desired properties, we extend it to the next discontinuity (or to $+\infty$ in case this was the last one) by considering the solution $(x^k(t), y^k(t),z^k(t))$ of the Cauchy problem on $[t_k, t_{k+1}]$ relative to the input function $u_k$ and initial condition $(x(t_k), y(t_k), z(t_k))$. We then extend $(x,y,z)$ to $[0, t_{k+1}]$ by setting $(x(t),y(t),z(t))=(x^k(t), y^k(t),z^k(t))$ for $t\in [t_k, t_{k+1}]$. This construction yields the first part of the proposition.

						Item (ii) is an immediate consequence of \eqref{eq:sol_y}.  The fact that $x(t)$ is non-decreasing follows from \eqref{X} and from the first equation in \eqref{control-system-full}. The fact that it is strictly decreasing if and only if $x_0>0$ and $y_0>0$ follows from item (ii), which ensures that $y(t)>0$ for every $t>0$ if and only if $y_0$, and from \eqref{eq:sol_x}, which ensures that $x(t)>0$ if and only if $x_0>0$. This proves item (i).
						
						To prove item (iii), note that point $(x_e, y_e, z_e)$ in $\Delta$ is an equilibrium if and only if 
						\begin{equation*}
							\begin{cases}
								0=-(1-u)\beta(x_e,y_e)x_ey_e,\\
								0=(1-u)\beta(x_e,y_e)x_ey_e-\gamma y_e,\\
								0=\gamma y_e\,,
							\end{cases}
						\end{equation*} 
						This implies that the points of the form $(x_e, 0, 1-x_e)$ are all equilibria and moreover, these are the only equilibria, since every point with $y_e \neq 0$ would not satisfy the equilibrium conditions given above. To prove item (iv), observe that the sum of the first two equations of \eqref{control-system-full} is $\dot{x}(t) + \dot{y}(t) = -\gamma y(t)\leq0$.
						Since $x(t)$ and $x(t) + y(t)$ are decreasing functions, they both admit a limit:
						$$x_\infty := \lim_{t \to +\infty} x(t),\;\; \xi_\infty := \lim_{t \to +\infty} [x(t)+y(t)]\,.$$
						As a consequence, the following limit
						\begin{equation}
							y_{\infty}= \lim_{t \to +\infty} y(t)=\xi_{\infty}-x_{\infty}\geq 0
						\end{equation}
						exists.	
						Suppose now by contradiction that $y_{\infty} >0$. This means that there exist $T, M >0$ such that $y(t) > M$, for all $t \geq T$. Then, we get 
						\begin{align*}
							\lim_{t \to +\infty} [x(t) + y(t) ]&= x(0) + y(0) - \gamma \int_{0}^{\infty} y(\tau) \de\tau \\
							&= x(0) + y(0) - \gamma \lim_{t \to +\infty} \int_{0}^{t} y(\tau) \de\tau \\
							&< x(0) + y(0) - \gamma \lim_{t \to +\infty} \int_{T}^{t} y(\tau) \de\tau \\
							&< x(0) + y(0) - \gamma M \lim_{t \to +\infty} (t- T) \\ 
							&= - \infty
						\end{align*}
						As this is a contradiction, it must be $y_\infty =0$. 
						This concludes the proof.

						\bibliographystyle{siamplain}
						\bibliography{bib}
					\end{document}